\documentclass[11pt, final]{amsart}
\usepackage{amsfonts}
\usepackage{amsmath}
\usepackage{amsthm}
\usepackage{amscd}
\usepackage{amsfonts}
\usepackage{amssymb}
\usepackage{mathrsfs}
\usepackage{graphicx, wrapfig}
\usepackage{enumitem}
\usepackage[font=scriptsize]{caption}

\usepackage[usenames]{color}
\definecolor{red}{rgb}{1.0,0.0,0.0}

\definecolor{blu}{rgb}{0.0,0.0,1.0}

\definecolor{gre}{rgb}{0.03,0.50,0.03}

\definecolor{amethyst}{rgb}{0.6, 0.4, 0.8}

\definecolor{blue-violet}{rgb}{0.54, 0.17, 0.89}

\definecolor{darkviolet}{rgb}{0.58, 0.0, 0.83}

\def\sqr#1#2{{\vcenter{\vbox{\hrule height .#2pt \hbox{\vrule
 width .#2pt height#1pt \kern#1pt \vrule
width .#2pt} \hrule height .#2pt}}}}

\def\qedo{\hbox{\hskip 6pt\vrule width6pt height7pt
depth1pt  \hskip1pt}\bigskip}

\def\ds{\begin{displaystyle}}
\def\eds{\end{displaystyle}}

\def\<{\left\langle }
\def\>{\right\rangle }

\def\H{\mathbb H}
\def\R{\mathbb R}

\def\calz{{\mathcal Z}}

\def\1{\mathbf 1}
\def\to{\rightarrow}

\newtheoremstyle{mytheorem}
{3pt}
{1pt}
{\itshape}
{-0pt}
{\bf}
{}
{1em}
{}
\newtheoremstyle{mydefinition}
{3pt}
{6pt}
{\itshape}
{-0pt}
{\bf}
{}
{1em}
{}
\newtheoremstyle{myremark}
{3pt}
{1pt}
{\rm}
{-0pt}
{\bf}
{}
{1em}
{}
\theoremstyle{mytheorem}
\newtheorem{Theorem}{Theorem}[section]
\newtheorem{Proposition}[Theorem]{Proposition}

\newtheorem{Corollary}[Theorem]{Corollary}

\theoremstyle{mydefinition}
\newtheorem{Definition}[Theorem]{Definition}
\newtheorem{Hypothesis}[Theorem]{Hypothesis}
\theoremstyle{myremark}
\newtheorem{Remark}[Theorem]{Remark}

\newtheorem{Notation}[Theorem]{Notation}

\usepackage[top=1.5in, bottom=1.7in, left=1.50in, right=1.50in]{geometry}

\DeclareMathOperator*{\argmax}{arg\,max}

\selectfont

\begin{document}

\title[Verification results for age-structured models...]{Verification results for age-structured models of economic-epidemics dynamics}

\author[G.Fabbri]{Giorgio Fabbri}
\address{Univ. Grenoble Alpes, CNRS, INRA, Grenoble INP, GAEL, Grenoble, France.}
\email{giorgio.fabbri@univ-grenoble-alpes.fr}

\author[F.Gozzi]{Fausto Gozzi}
\address{Dipartment of Economics and Finance, LUISS University, Rome, Italy.}
\email{fgozzi@luiss.it}

\author[G. Zanco]{Giovanni Zanco}
\address{Dipartment of Economics and Finance, LUISS University, Rome, Italy.}
\email{gzanco@luiss.it}


\begin{abstract}

In this paper we propose a macro-dynamic age-structured set-up for the analysis of epidemics/economic dynamics in continuous time.

\smallskip

The resulting optimal control problem is reformulated in an infinite dimensional Hilbert space framework where we perform the basic steps of dynamic programming approach.

\smallskip

Our main result is a verification theorem which allows to guess the feedback form of optimal strategies. This will be a departure point to discuss the behavior of the models of the family we introduce and their policy implications.
	



\end{abstract}

\maketitle


\section{Introduction}

The outbreak of the COVID-19 pandemic represents, in addition to an epidemiological historical event, an exceptional economic shock. Data from the OECD (2020) suggest that in many countries the loss of GDP due to the presence of the virus and the consequent containment measures will be at least 10\%. For this reason, together with the obvious upsurge in medical scientific production on the subject, the phenomenon has had a great echo in economic literature with a strong pressure to merge economic and epidemiological models.

\medskip

A specific effort has been made to integrate epidemiological compartmental models (SIR, SEIR, SEI...) into a macroeconomic dynamic context, see for example the contributions of Alvarez et al., (2020), Eichenbaum et al., (2020), Jones et al., (2020) and Krueger et al.(2020).

These articles focus on a series of questions essential to health and economic policy and they look, often numerically, at the trade-off between measures capable of containing contagion and those capable of avoiding economic collapse. However, they model the spread of the epidemic with age homogeneous epidemiological compartmental models so they cannot take into account  one of the characteristic traits of the current epidemic, i.e. the great difference in the effects of the disease among people of different ages\footnote{In fact, the probability of aggravation of COVID-19 infection and mortality varies very significantly with age. Salje et al. (2020) find for example that less than 1\% of people under 40 years of age who contract the disease need hospital care against more than 10\% of people over 70 years of age and that mortality in the two groups is respectively less than 0.02\% and more than 2\%.}.

\medskip

In order to address this limitation Acemoglu et al. (2020), Gollier (2020) and Favero et al. (2020) introduce models where the population is divided into a finite number of homogeneous ``risk groups'' and they study joint economic and epidemiological effects of introducing group-specific policies. Nonetheless in their formulations there is no possibility to move from one group to another and then this kind of approach can take into account the different effects of the disease on different age groups only if it is assumed that the duration of the epidemic is negligible compared to the age range contained in each group. However, this hypothesis is not very likely in the case of an epidemic lasting several years and it is inadequate in the case of diseases that become endemic in the population\footnote{These limitations are obviously justified by the need to produce policy indications in a short time in order to contrast the spread of the current pandemic.}.

\medskip

Instead of using age-homogeneous epidemiological compartmental models or epidemiological compartmental models with closed risk groups, it is possible, as we do in the present work, to describe more accurately the joint dynamics of the epidemic and of the age structure of the population by using explicit age-structured compartmental models, i.e. age-specific epidemiological models with ageing process modeled \emph{\`a la} Mc Kendrick (1925). This type of models was initially introduced by Anderson and May (1985) and Dietz and Schenzle (1985) and later adapted to numerous contexts and applications, see the books by Iannelli (1995), Iannelli and Milner (2017) and Martcheva (2015) for a structured and modern description of the matter.

The more accurate description of the ageing-epidemics diffusion dynamics comes at a price and, indeed, one of the features of the continuous time compartmental age-structured models is to describe the epidemiological dynamics through transport type partial differential equations  (PDEs). This means that, if one wants to study an associated optimal control problem through dynamic programming, its dynamics needs to be seen in an infinite-dimensional set-up.

\bigskip

In this paper we present verification type results for a class of macro-dynamic models that incorporates an epidemiological dynamics which generalizes the benchmark age-structured SIR model.

In the model, the planner optimizes an optimal social functional by choosing policies to reduce the spread of the virus, taking into account (i) their effectiveness in terms of human lives saved  (ii) their impact on the labor supply and therefore on the level of production and consumption reached by the country (iii) their direct costs. As in standard growth models, the planner also decides the level of consumption, which can be age-specific, and consequently, the level of investments that determines the dynamics of capital accumulation.

The class of models that we study in the abstract form is rather general and is able, in the context of the epidemiological dynamic described by an age-structured SIR, to reproduce as special cases several of the settings proposed by the recent articles mentioned above (more details can be found in Section \ref{sec:model}). Specific traits of the model are:
\begin{itemize}
\item[-] The epidemiological model is general and can be set with a wide variety of age-dependent parameters: mortality rate (due to epidemics and also to other causes), chance of being hospitalized if infected, birth rate and probability of contagion among cohorts. Moreover, the age-specific mortality rate can take into account the saturation of hospitals and health systems, a phenomenon that has been repeatedly observed in the areas most affected by COVID-19 (see for instance Moghadas, 2020)
\item[-] The planner has two different policy levers: on the one hand, as in most of the models mentioned, she can reduce the mobility of people and partially stop the economic activity (lock-down), on the other hand she can implement some costly action to reduce the diffusion of the virus, for instance by testing the population extensively to try to quarantine individuals fast once they have contracted the virus. Both policies can be  age-specific (in particular targeted lock-downs suggested by Acemoglu et al., 2020 are among possible policies)
\item[-] The time horizon can be either infinite (as in benchmark growth models) or finite if it is considered (e.g. Gollier, 2020) that the spread of the virus stops at some point due to the discovery of a vaccine or a cure
\item[-] Labor productivity is age-specific (this fact is important for policies: targeted lock-downs for less productive people impact less the production). The production function (as a function of aggregated labor and capital) is general as well as the optimal social function that can be specified to take into account cost-benefit analysis, strictly humanitarian or economic targets, standard (Benthamite and Millian for instance) social welfare functionals.
\end{itemize}
Since our results are proven for the abstract model they hold for any possible specification.

The contribution of this work is (i) to propose a general fully age-structured macro-dynamic set-up in continuous time for analysis of epidemics and economic dynamic (Section \ref{sec:model}); (ii) to provide a suitable Hilbert space environment where one can rewrite the problem and perform dynamic programming (Section \ref{sec:infdim}); (iii) to prove verification type results (Section \ref{sec:verification}), see in particular Theorem \ref{thm:comparison} and Corollaries \ref{cor:verification1} and \ref{cor:verification2}.

We must be clear on the fact that we do not solve the problem explicitly, nor numerically. Here our main goal is to provide a general ground which can be the departure point to attack special cases of our general model. In particular our main contribution is the proof of the verification type results of Section \ref{sec:verification}. These are
nontrivial to obtain in our general infinite dimensional setting and they are crucial to find the optimal policies in a closed-loop form depending on the derivatives of the value function.
These type of theorems are the object of various papers (see e.g. the papers Faggian and Gozzi (2010), Fabbri et al. (2010)) or of book chapters (e.g. Chapter 5 of Yong and Zhou, 1999 or Chapter 4 of Li and Yong, 1995) but none of therm applies to our case.
The main reasons are the following.
First, due to the age-structured nature of the problem and the presence of the mortality forces, we have to work with semigroups in weighted infinite dimensional spaces which do not have regularizing properties (which are very useful and which are usually true when the state equation is of (nondegenerate) second order).
Second, the presence of the nonlinear equation for capital rules out the standard regularity assumptions that are used e.g. in Chapter 4 of Li and Yong (1995) and that we treat using \emph{ad-hoc} arguments.
Third, in our case we have state constraints, which makes much more difficult to deal with the problem. We use the approach of weakening the constraints which has been used, up to now, only in case when explicit solutions of the HJB equation are available (see e.g. Fabbri and Gozzi, 2008, Boucekkine et al., 2019).



\bigskip

The paper is organized as follows. In Section \ref{sec:model} we introduce the structure of the model: epidemiological dynamics, policies, structure of the economy and welfare functional. In Section \ref{sec:infdim} we show how to reformulate the model and the related optimal control problem in a suitable Hilbert space setting. Section \ref{sec:HJB} is devoted to dynamic programming while in Section \ref{sec:verification} we provide the verification results. Section \ref{sec:conclusion} concludes.

\section{The model}
\label{sec:model}

\subsection{Epidemics dynamics}

We denote by $s(a,t)$ the density of susceptible individuals of age $a\in[0, \bar a]$ (being $\bar a>0$ the maximum age) at time $t\geq 0$. Similarly $i(a,t)$ (respectively $r(a,t)$) denotes the density of infected/infectious individuals (respectively recovered individuals) of age $a$ at time $t$.
Hence the total numbers of susceptible, infected, and recovered individuals are
\[
S(t) = \int_0^{\bar a} s(a,t) da,
\qquad I(t) = \int_0^{\bar a} i(a,t) da,
\qquad R(t) = \int_0^{\bar a} r(a,t) da.
\]
The age-dependent density of the total population $n(a,t)$ is then given by
\[
n(a,t) = s(a,t) + i(a,t) + r(a,t)
\]
and the total population is
\[
N(t) = \int_0^{\bar a} n(a,t) da = S(t) + I(t) + R(t).
\]
In modeling the mortality we generalize the standard age-structure SIR framework (see Martcheva, 2015, Chapter 12).
First we define
\[
\Xi(t) := \int_0^{\bar a} i(a,t) \xi(a) da
\]
the number of people at time $t$ in ``critical conditions'' i.e. people who have to use the services of hospital/healthcare facilities to treat themselves at the risk of saturating them. In the case of the COVID-19 epidemic, the emphasis is, for example, on people needing to be hospitalized in an ICU (see for instance Moghadas, 2020). $\Xi(t)$ depends on the number of sick people per cohort multiplied by the prevalence $\xi(a)$ of people in need of specific care for each age group\footnote{In the case of COVID-19 for example, in the data of Salje (2020), 2.9\% of infected individuals are hospitalized ranging from 0.1\% in people under 20 years to approximately 30\% in individuals with 80 years of age or older.}.
In a context of saturation of hospital services, the mortality of the infected will be increased. Let us therefore assume that the mortality rate for infected individual $\mu_I$ is not only a function of the age of the individuals but it also en increasing function of $\Xi$. We use then the notation $\mu_I(a, \Xi(t))$. We suppose, for simplicity, that the (age-specific) mortality rates of susceptible and recovered individuals, respectively $\mu_S(a)$ and $\mu_R(a)$ do not depend\footnote{Indeed it is possible to incorporate this dependence in the model without big problems.} on $\Xi(t)$. Finally $\gamma(a)$ and $\beta(a)$ denote respectively the (age-specific) recovery and birth rates.

The age-specific force of infection $\lambda (a,t)$  depends on the distribution of infected individuals as follows
  \begin{equation}\label{eq:lambdadef}
    \lambda(a,t)=\frac{1}{N(t)}\int_0^{\bar a}m(a,\tau) i(\tau,t) d\tau.
  \end{equation}
In this expression the joint-distribution $m(a,\tau)$ measures the different probability of contagion between cohorts (for instance virus diffusion can be easier among children for childhood diseases). It is the continuous version of the social contact matrix across age classes used by Gollier (2020).


All in all, the \emph{laissez faire} benchmark population dynamics, that is the epidemics dynamics without policy intervention (omitting the initial conditions at time $t=0$) is the following:
\begin{equation}
\label{eq:state-freediffusion}
\left \{
\begin{array}{l}
\frac{\partial s(a,t)}{\partial t} + \frac{\partial s(a,t)}{\partial a} = - \lambda(a,t)s(a,t) - \mu_S (a)s(a,t),\\[5pt]
\frac{\partial i(a,t)}{\partial t} + \frac{\partial i(a,t)}{\partial a} =  \lambda(a,t)s(a,t) - (\mu_I (a, \Xi(t)) + \gamma (a))i(a,t),\\[5pt]
\frac{\partial r(a,t)}{\partial t} + \frac{\partial r(a,t)}{\partial a} =\gamma (a)i(a,t) - \mu_R (a) r(a,t)\\[5pt]
s(0,t) = \int_{0}^{\bar a} \beta(a) n(a,t) da\\[5pt]
i(0,t) = 0 \\[5pt]
r(0,t) = 0.
\end{array}
\right .
\end{equation}
This system is the standard age-structure SIR model (see Martcheva, 2015, Chapter 12) except for the fact that $\mu_I$ depend on $\Xi$. In the particular case where $\mu_I(a,\Xi(t))=\tilde \mu_I(a)$ we are exactly in the standard setting.
Note that, since $\lambda$ and $\Xi$ depend linearly on $i$, the system \eqref{eq:state-freediffusion} is non linear in the variables $(s,i,r)$.

\medskip

We now introduce two of the three policy levers that the planner has in our model (the third is the choice of consumption and will be described in the next subsection). We suppose that the planner can deal with the epidemic in two ways:
\begin{itemize}
\item[(i)] partially stopping economic activity and people mobility then reducing the contagion frequency among individual (lockdown);
\item[(ii)] implementing some costly action to reduce the diffusion of the virus, for instance by testing the population extensively to try to quarantine individuals faster once they have contracted the virus.
\end{itemize}
More precisely
\begin{itemize}
	\item[(i)] We suppose that the planner can reduce mobility and then the probability of infecting and being infected of cohort $a$ at time $t$ by a factor $\theta(a,t)\in [0,1]$ at the cost of reducing the contribution of the concerned individuals to work or by reducing their work productivity (for example resorting to teleworking). 
	This is the type of intervention which is modeled in almost all the macro-dynamic models we mentioned in the introduction, for instance in Alvarez et al., (2020) and Eichenbaum et al., (2020) where, by the way, age-structure policies are not possible since there is no age structure of the population. Taking different values of $\theta$ for different $a$ correspond to target lock-downs of Acemoglu et al. (2020). 
	\item[(ii)] We suppose, as in some of the mentioned papers, that the planner can reduce by a factor $\eta(a,t)\in [0,1]$ the probability that infected individuals of cohort $a$ at time $t$ contaminate other people. This is done at the cost
\begin{equation}\label{eq:deftestcost}
	D_\eta(t) := D\left ( \int_0^{\bar a} \eta(a,t) i(a,t)  e(a)  da \right ).
\end{equation}
	where $e(a)$ is an age-specific relative cost and $D$ is a concave (as, for instance in Piguillem and Shi, 2020) or linear (as in Gollier, 2020) function which represents some form of congestion (e.g. shortage of tests on the international market or shortage of suitable medical personnel to administer the tests).
\end{itemize}

The evolution of the epidemics is then again described by (\ref{eq:state-freediffusion}) but, instead of $\lambda(a,t)$ written in \eqref{eq:lambdadef} we have now the following age-specific force
\begin{equation}\label{eq:lambdadefcontrolled}
    \lambda^{\theta,\eta}(a,t) = \frac{\theta(a,t)}{N(t)} \int_0^{\bar a} m(a,\tau) \theta(\tau,t) \eta(\tau,t) i(\tau, t) d\tau.
\end{equation}
Hence we get the following state equation (still omitting the initial conditions at time $t=0$):
\begin{equation}
\label{eq:state-controlled}
\left \{
\begin{array}{l}
\frac{\partial s(a,t)}{\partial t} + \frac{\partial s(a,t)}{\partial a} = - \lambda^{\theta,\eta}(a,t)s(a,t) -\mu_S (a)s(a,t),\\[5pt]
\frac{\partial i(a,t)}{\partial t} + \frac{\partial i(a,t)}{\partial a} =  \lambda^{\theta,\eta}(a,t)s(a,t) - (\mu_I (a, \Xi(t)) + \gamma (a))i(a,t),\\[5pt]
\frac{\partial r(a,t)}{\partial t} + \frac{\partial r(a,t)}{\partial a} =\gamma (a)i(a,t) - \mu_R (a) r(a,t)\\[5pt]
s(0,t) = \int_{0}^{\bar a} \beta(a) n(a,t) da\\[5pt]
i(0,t) = 0 \\[5pt]
r(0,t) = 0.
\end{array}
\right .
\end{equation}
Of course if the authority fixes $\theta(a,t)\equiv 1$, $\eta(a,t)\equiv 1$ we find again the free diffusion dynamics (\ref{eq:state-freediffusion}).

\medskip

\subsection{Production and capital accumulation}
We suppose that labor supply is perfectly inelastic to wage, that infected people do not work and that labor productivity is age-specific\footnote{We abstract from other reasons of productivity heterogeneity among population and from heterogeneity of tasks.} and proportional to a certain parameter $\alpha(a)$ (we can specify for instance $\alpha(a)=0$ for children or for individuals older than a fixed retirement age). Total labor supply in efficiency units in the \emph{laissez faire} benchmark is then given by $\int_0^{\bar a} (s(a,t) + r(a,t)) \alpha (a) da$. In the controlled case we suppose that getting a factor $\theta(a,t)\in [0,1]$ in the expression of the age-specific force of diffusion impacts the productivity of cohort $a$ reducing the productivity to $\varphi(\theta(a,t))$ so that total labor supply in efficiency units\footnote{A similar approach is considered for instance by Jones et al. (2020) which introduce an ``effective labor supply''.} is now
\begin{equation}
\label{eq:laboursupply}
L(t) = \int_0^{\bar a} (s(a,t) + r(a,t)) \alpha (a) \varphi(\theta(a,t)) da.
\end{equation}
We suppose that $\varphi\colon [0,1] \to [0,1]$ is an increasing function with $\varphi(1)=1$.

As for the production we stick to the standard structure of neoclassical growth models and we suppose that the total production at time $t$ is described by an aggregated production function $F$ of the two factors: labor $L(t)$ and capital $K(t)$:
\[
Y(t) = F(K(t), L(t)).
\]
This formulation is more general than that used by other macro-dynamic papers we mentioned. Indeed in all of them except Favero et al. (2020) which uses a Cobb-Douglas production function, the authors use production functions which are linear function of the of labor (or effective labor) or even do not model production.



We abstract from international trade (closed economy) and from governmental expenditure so the planner can choose at any time $t\ge 0$ how to allocate the national total production $Y(t)$ among total investment $M(t)$, consumption of various cohorts and costs for testing people, which is defined in \eqref{eq:deftestcost} above. If we denote by $c(a,t)$ the per-capita consumption of individuals of age $a$ at time $t$ we get the following budget constraint:
\[
Y(t) = M(t) + C(t) + D_\eta(t) := M(t) + \int_0^{\bar a} c(a,t) n(a,t) da + D\left ( \int_0^{\bar a} \eta(a,t) i(a,t)  e(a) da  \right ).
\]
Supposing to have an exponential capital depreciation \emph{\`a la} Jorgenson we get the dynamic accumulation law for capital:
\begin{equation}
  \label{eq:K}
\dot{K}(t) = F(K(t),L(t)) - \int_0^{\bar a} c(a,t) n(a,t) da \, -\delta K(t) - D\left ( \int_0^{\bar a} \eta(a,t) i(a,t)  e(a) da  \right ).
\end{equation}
where $\delta>0$ is the constant depreciation rate.

Observe that, as far as we know (but the litterature is quickly growing) this is the first paper among the macro-dynamic papers on COVID-19 pandemics where capital accumulation is explicitly taken into account (as said before almost all the papers we mentioned even consider labor as unique production factor). This is of course very relevant if one wants to understand the consequences in terms of investments of the epidemiological shock.


\subsection{Choosing the target}
%
%
%

For the functional to maximize there are several interesting choices in the literature. It is not easy to include them all in an abstract form that leaves the problem tractable, so in this section we introduce several functionals that will be discussed later in the article.

The first functional we introduce is a standard welfare functional. Observe that, even if the model we study here is not directly an endogenous fertility model, the fact of having an endogenous mortality (depending on the choice of $\theta$) makes it \emph{de facto} an endogenous population size model. Therefore we have to choose carefully the structure of the
social utility that we describe. We implicitly fix the utility of dead people (and non-born people through the initial condition $\int_{0}^{\bar a} \beta(a) n(a,t) da$) equal to $0$ and we consider the following social utility functional:
\begin{equation}\label{eq:target1infhor}
\int_0^\infty \int_0^{\bar a} e^{-\rho t} n^\nu(a,t) u(c(a,t),\theta(a,t)) \, da \, dt.
\end{equation}
To assure, for the same per capita (age-dependent) consumption, the instantaneous utility to be increasing in the number of living people and therefore the planner being averse to death of agents, the per-capita utility function $u$ needs to be positive.
Still observe that in this model formulation there is room for a dilution effect: the larger the population the lower the percapita consumption so the instantaneous utility does not need to always be increasing in the population size.

The per-capita utility function $u$ depends both on the individual consumption and on the mobility freedom $\theta$. We suppose that $u$ is (positive and) an increasing function in both the variables. The dependence of utility on $\theta$ is not standard but the relevance of this choice can easily be argued by looking at the various side effects of lock down (see for example Clemens, 2020). In any case, as a special case, of course one can specify $u$ so that it does not depend on theta.

The form of this first functional is the age-structured version of a standard functional often appearing in the optimal population literature. The parameter $\nu$ which appears in its expression measures the degree of altruism towards individuals of future cohorts (see Palivos and Yip, 1993). The case $\nu = 1$ corresponds to the classical total utilitaristic (or ``Benthamite'') case where the planner target is to maximize the sum of individuals' utility.


The functional (\ref{eq:target1infhor}) is infinite horizon and implicitly suggests that no exogenous element impedes the spread of the virus. Another possibility, as suggested by Gollier (2020), is to consider a final time $T$ at which an event (a cure or more probably the discovery of a vaccine) stops the epidemics. We describe some possible targets in this context.

The trade-off of virus containment policies is: reducing the number of deaths VS economic losses. Some of the functional aspects can be dwelt on only one of these aspects. For instance one can decide to focus on economic activity and to maximize the final production capacity:
\begin{equation}\label{eq:target1finhor}
F(K(T),\tilde L(T))
\end{equation}
where $\tilde L(T)$ is defined as
\begin{equation}\label{eq:tildeLdef}
\tilde L(T) = \int_0^{\bar a} n(a,T) \alpha (a) da
\end{equation}
(once the vaccine is in and the outbreak is over, everyone is cured and everyone is productive) or even more simply, to maximize final capital level
\begin{equation}
\label{eq:target2finhor}
K(T)
\end{equation}
or to maximize the flow of production
\begin{equation}
\label{eq:target2infhor}
\int_0^{+\infty} e^{-\rho t} Y(t) dt.
\end{equation}
or its finite counterpart
\begin{equation}\label{eq:target3finhor}
\int_0^{T} e^{-\rho t} Y(t) dt.
\end{equation}
Conversely one can focus on humanitarian aspects and decide to maximize the number of deaths due to the virus:
\begin{equation}\label{eq:target4finhor}
\int_0^{T} \int_0^{\bar a} \mu_I (a, \Xi(t)) i(a,t) da dt.
\end{equation}

It is also possible to simply take into account both the economic and humanitarian aspects by taking a weighted sum of the (\ref{eq:target2infhor}) (or \ref{eq:target2finhor}) or (\ref{eq:target3finhor})) and (\ref{eq:target1finhor}), this is the choice of Acemoglu et al. (2020).

\section{Infinite dimensional formulation of the model}
\label{sec:infdim}


In this section we introduce a convenient infinite-dimensional formulation for system (\ref{eq:state-controlled}) coupled with equation (\ref{eq:K}) and for the control problem of maximizing the target given in \eqref{eq:target1infhor} (or in \eqref{eq:target2infhor}, \eqref{eq:target1finhor}, \eqref{eq:target2finhor}, \eqref{eq:target3finhor}).
We make the following set of assumptions, which also includes those already stated in the previous sections.
These assumptions will be always true in the remainder of the paper
without mentioning them.
\begin{Hypothesis}\label{hp:main}
  \begin{enumerate}[label=$(\roman{*})$]
\item[]
\item $\mu_S$ and $\mu_R$ are positive, belong to $L^1_{\text{loc}}(0,\bar a)$ and
    $$
    \int_0^{\bar a}\mu_S(a)da=\int_0^{\bar a}\mu_R(a)da=+\infty;
    $$
\item $\mu_I\colon[0,\bar a]\times \R\to \mathbb{R}_+$ is measurable. Moreoves it is Lipschitz continuous in the second variable, uniformly with respect to the first one. Finally it is increasing in the second variable and
        $$
    \int_0^{\bar a}\mu_I(a,\kappa)da=+\infty, \qquad \forall \kappa\in \R;
    $$
\item $F(\cdot,L)$ is Lipschitz for every $L\in\mathbb{R_+}$, with Lipschitz constants uniformly bounded in $L$\footnote{{This assumption is not verified for the Cobb-Douglas type functions which, however, can be treated ad hoc in this framework.}};
\item $\varphi\colon[0,1]\to[0,1]$ is increasing and $\varphi(1)=1$;
\item $\alpha,\beta,\gamma,e,\xi\colon[0,\bar a]\to \mathbb{R}_+$ are in $L^2(0,\bar a)$;
  \item $D:\mathbb{R}\to\mathbb{R}$ is positive and concave;
  \item $\delta>0$, $\nu \in [0,1]$;
    \item $u\colon\mathbb{R}\times[0,1]\to\mathbb{R}$ is positive, continuous and increasing in both variables;
  \end{enumerate}
\end{Hypothesis}

We now start rewriting the system (\ref{eq:state-controlled}) but first we introduce an important notational standard.

\begin{Notation}\label{not:infdim}
In the system (\ref{eq:state-controlled}) the three state trajectories, $s(\cdot,\cdot)$, $i(\cdot,\cdot)$, $r(\cdot,\cdot)$ are seen as function of two variables, i.e.
$$
(s,i,r)(\cdot,\cdot):[0,+\infty) \times [0,\bar a]\to \R^3_+,
\qquad (t,a) \to \left(s(t,a),i(t,a),r(t,a)\right)
$$
However now it is convenient to see such trajectories as functions from
$t \in \R_+$ to a suitable infinite dimensional Hilbert space $H$ of functions in the variable $a\in [0,\bar a]$ with values in $\R^3$. $H$ can be seen also as the product of three Hilbert spaces of functions with values in $\R$ and its generic element will be denoted by $h=(h_1,h_2,h_3)$ or, if no confusion is possible, by $(s,i,r)$. To avoid misunderstandings we will denote the state trajectories putting a hat over the original name, i.e.
$$
(\hat s,\hat i,\hat r): \R_+ \to H
\qquad t \to \left(\hat s(t),\hat i(t),\hat r(t)\right)
$$
Sometimes we write $\hat h$ for $(\hat s,\hat i,\hat r)$
and, when we want to underline that they are functions we write
$\hat h(\cdot)$ or $(\hat s(\cdot),\hat i(\cdot),\hat r(\cdot))$.
Now observe that, for every $t\ge 0$,
$\hat s(t),\hat i(t),\hat r(t)$ are functions of $a$.
We will denote their value at a given $a\in [0,\bar a]$ with
$\hat s(t)[a],\hat i(t)[a],\hat r(t)[a]$ so to emphasize the different role of the two variables. Clearly we will have
$$
\hat s(t)[a]= s(t,a), \quad \hat i(t)[a]= i(t,a),
\quad \hat r(t)[a]= r(t,a).
$$
The same will be done for the controls strategies
$c(\cdot,\cdot)$, $\theta(\cdot,\cdot)$, $\eta(\cdot,\cdot)$.
More precisely we will fix a control space $Z$ of functions in the variable $a\in [0,\bar a]$ with values in $\R^3$. Also $Z$ can be seen as the product of three Hilbert spaces of functions with values in $\R$ and its generic element will be denoted by $z=(z_1,z_2,z_3)$ or, if no confusion is possible, by $(c,\theta,\eta)$. Also here, to avoid misunderstandings, we will denote the control trajectories putting a hat over the original name, i.e. we call
$\hat c, \hat \theta, \hat \eta$ the functions
$$
(\hat c,\hat \theta,\hat \eta): \R_+ \to Z
\qquad t \to \left(\hat c(t),\hat \theta(t),\hat \eta(t)\right)
$$
such that
$$
\hat c(t)[a]= c(t,a), \quad \hat \theta(t)[a]= \theta(t,a),
\quad \hat \eta(t)[a]= \eta(t,a).
$$
Sometimes we write $\hat z$ for $(\hat c,\hat \theta,\hat \eta)$
and, when we want to underline that they are functions we write
$\hat z(\cdot)$ or $(\hat c(\cdot),\hat \theta(\cdot),\hat \eta(\cdot))$.
\hfill\qedo
\end{Notation}

We are now ready to introduce the spaces $H$, $Z$, and the space
$\calz_0$ of basic control strategies (this is not the space of admissible control strategy which will have to take account of the state constraints that we will introduce below).

\color{black}

Define the probability of surviving to age $a$ for a susceptible individual as
\begin{equation*}
  \pi_S(a)=\exp\left(-\int_0^a\mu_S(\tau)d\tau\right)
\end{equation*}
and, similarly, define the probability of surviving to age $a$ for a recovered individual as
\begin{equation*}
  \pi_R(a)=\exp\left(-\int_0^a\mu_R(\tau)d\tau\right)\ .
\end{equation*}
Consider the set
\begin{equation*}
  H=\left\{h\in L^2(0,\bar a;\R^3)\colon \frac{h_1}{\pi_S}\in L^2(0,\bar a),h_2\in L^2(0,\bar a), \frac{h_3}{\pi_R}
  \in L^2(0,\bar a)\right\}\; .
  \end{equation*}
$H$ is a Hilbert space when endowed with the inner product
\begin{equation*}
  \begin{aligned}
    \langle h,g\rangle_H &= \langle \frac{h_1}{\pi_S},\frac{g_1}{\pi_S}\rangle_{L^2}+\langle h_2,g_2\rangle_{L^2}+\langle\frac{h_3}{\pi_R},\frac{g_3}{\pi_R}\rangle_{L^2}\\
    &=:\langle h_1,g_1\rangle_{\pi_S}+\langle h_2,g_2\rangle_{L^2}+\langle h_3,g_3\rangle_{\pi_R}\ .
  \end{aligned}
\end{equation*}
\begin{Remark}
The choice of the space $H$ is different from the standard one made, e.g., in Iannelli and Martcheva, 2003. Indeed it would be standard to put the weight also on the second component. however this is not possible since, in our model, we have the new and important feature that the mortality force $\mu_I$ is state dependent. With our choice the space $H$ is bigger than the usual one but it is still possible to formulate the problem there. We finally observe that this choice will reflect also in the form of the adjoint operator in Proposition \ref{pr:SE}.  \hfill \qedo
\end{Remark}

It is useful, for later purpose to introduce the positive cone in $H$ as follows
\begin{equation*}
  H_+=\left\{h\in H\colon h_i\geq 0\text{ a.e. in }[0,\bar{a}]\right\}\;\subset \; H.
\end{equation*}
The control space $Z$ is given as:
\begin{equation*}
Z=\left\{z=(z_1,z_2,z_3)\colon\, z_i \in L^2(0,\bar a), i=1,2,3; z_1(a)\ge 0,\, z_2(a),\, z_3(a)\in [0,1],\; \forall a \in[0,\bar a]\right\}\; .
  \end{equation*}
$H$ is a Hilbert space when endowed with the inner product
\begin{equation*}
  \begin{aligned}
    \langle z,w\rangle_Z &= \langle z_1,w_1\rangle_{L^2}+\langle z_2,w_2\rangle_{L^2}+\langle z_3,w_3\rangle_{L^2}.
  \end{aligned}
\end{equation*}
Finally the space $\calz_0$ of all basic control strategies is, coherently with the requirements of Section 2, a space of functions from $\R_+$ to $Z$ and is given as follows
$$
\calz_0:= L^2(\mathbb{R}_+;Z);
$$
For coherence with the Notation \ref{not:infdim}
introduced above we will call $z=(c,\theta,\eta)$ the points of $Z$.
and $\hat z=(\hat c,\hat\theta,\hat\eta)$ the points of $\calz_0$.

Now we reformulate system (\ref{eq:state-controlled}) providing also existence and uniqueness of the solution. We need to introduce some operators which comes from the various term of the system.

First we introduce the unbounded linear operator $A:D(A)\subset H\to H$ defined as
\begin{equation*}
  A=\begin{pmatrix} -\frac{\partial}{\partial a}-\mu_S&0&0\\0&-\frac{\partial}{\partial a}-\gamma&0\\0&\gamma&-\frac{\partial}{\partial a}-\mu_R\end{pmatrix}
\end{equation*}
with domain
\begin{multline*}
  D(A)=\left\{h\in H\colon \frac{h_1}{\pi_S},h_2,\frac{h_3}{\pi_R}\in W^{1,2}(0,\bar{a}),\right.\\ \left.h_1(0)=\int_0^{\bar a}\beta(a)(h_1+h_2+h_3)(a)da, h_2(0)=h_3(0)=0\right\}\ ,
\end{multline*}
corresponding to the linear part of system \eqref{eq:state-controlled}.
It can be shown as in Iannelli and Martcheva (2003) that $A$ generates a strongly continuous semigroup $T(t)$ on $H$ such that $T(t)(H_+)\subset(H_+)$ for every $t\geq 0$.\\

Second we define linear functional $\bar \Xi$ which reformulates the function $\Xi$ given in Section 2.
\begin{gather}\label{eq:Xiop}
\bar \Xi\colon H \to \mathbb{R}\\
  \bar \Xi(h)=\int_0^{\bar a} h_2(a)\xi(a)da.
\end{gather}
For every control point $z=(c,\theta,\eta)\in Z$ we define the nonlinear operators (depending only on the components $\theta$ and $\eta$ of the control point)
\begin{gather*}
  \Lambda^{\theta,\eta}\colon H_+\setminus\{0\}\to L^\infty(0,\bar a)\\
  \Lambda^{\theta,\eta}(h)(a)=\frac{\theta(a)}{\int_0^{\bar a}(h_1+h_2+h_3)(\tau)d \tau} \int_0^{\bar a}m(a,\tau)\theta(\tau)\eta(\tau)h_2(\tau)d\tau
\end{gather*}
and
$$B^{\theta,\eta}:H_+\setminus\{0\}\to H,$$
\begin{equation*}
  B^{\theta, \eta}\left(h\right)(a)=\begin{pmatrix} -\Lambda^{\theta,\eta}(h)(a)h_1(a)\\ \Lambda^{\theta,\eta}(h)h_1(a)-\mu_I(a,\Xi(h))h_2(a)\\0\end{pmatrix}\ .
\end{equation*}
Now by Hypothesis \ref{hp:main}(ii) and by the fact that (see the definition of $Z$) we have $\theta,\eta\in L^\infty(0,\bar a)$, the operator $B^{\theta,\eta}$ is
Lipschitz continuous on $H_+\setminus\{0\}$ and there exists a positive constant $\alpha$ such that $\alpha B^{\theta,\eta}(h)+h\in H_+$ for every $h\in H_+$ (see Iannelli and Martcheva, 2003).

We can consequently write system (\ref{eq:state-controlled}) as the evolution equation for the unknown $\hat h\colon[0,+\infty)\to D(A)$
with control strategies $\hat \theta$ and $\hat \eta$:
\begin{equation}
  \label{eq:ODE}
  \frac{d}{dt}\hat h(t)=A\hat h(t)+B^{\hat \theta(t),\hat \eta(t)}(\hat h(t))\ .
\end{equation}
Given control strategies $\hat\theta(\cdot),\hat\eta(\cdot)$ and an initial condition $h_0\in H_+$ we look for mild solutions of the above systems in $H_+$, i.e., for functions $[0,+\infty)\ni t\mapsto \hat h(t)\in H_+$ that satisfy
\begin{equation*}
  \hat h(t)=T(t) h_0+\int_0^tT(t-s)B^{\hat\theta(s),\hat\eta(s)}(\hat h(s)) ds
  \end{equation*}
Thanks to the fact that $A$ generates a strongly continuous semigroup that leaves $H_+$ invariant and thanks to the properties of $B$, there exists (see e.g. Bensoussan et al. (2007) a unique function $\hat h(\cdot)$ that satisfies (\ref{eq:ODE}) and such that $\hat h(0)=h_0$ and $\hat h(t)\in H_+$ for every $t\in[0,+\infty)$. Such solution will be denoted by $\hat h^{\hat\theta,\hat\eta;h_0}$
or by $\hat h^{\hat z;h_0}$.

We now add the equation for $K$ to the system, see (\ref{eq:K}). For
control points $z=(c,\theta, \eta)\in Z$
we define the functionals on $H$
\begin{equation*}
  L^{\theta}(h)=\int_0^{\bar a}\left( h_1(a)+h_3(a)\right)\alpha(a)\varphi(\theta(a))da,
\end{equation*}
\begin{equation*}
  C^{c}(h)=\int_0^{\bar a}c(a)(h_1+h_2+h_3)(a)da,
\end{equation*}
\begin{equation*}
  D^{\eta}(h)=D\left(\int_0^{\bar a}\eta(a)h_2(a)e(a)da\right).
\end{equation*}
For any control strategy $\hat z(\cdot)\in \calz_0$, any $h_0\in H_+$ and any $K_0\in\mathbb{R}$ we are then considering the Cauchy problem
\begin{equation}
  \label{eq:hK1}
  \begin{cases}
    \hat h'(t)&=A\hat h(t)+B^{\hat\theta(t),\hat\eta(t)}(\hat h(t))\\
    K'(t)&=-\delta K+F\left(K(t),L^{\hat\theta(t)}(\hat h(t))\right)-C^{\hat c(t)}(\hat h(t))-D^{\hat\eta(t)}(\hat h(t)),\\
    \hat h(0)&=h_0,\\
    K(0)&=K_0.
   \end{cases}
\end{equation}
Observe that the first equation does not depend on $K$. Hence, once we know the mild solution $h^{\hat z;h_0}$ of the first equation we can plug it into the second one. Since $F$ is Lipschitz in $K$, uniformly in the second variable\footnote{The fact the two equations are not fully coupled can be exploited to cover also the case when $F$ is a Cobb-Douglas function by using a Bernoulli-type change of variable.}, we know that the second equation has a solution $K^{\hat z;K_0,h_0}$ (note that it depends also on $h_0$ since the trajectory $\hat h$ appears in the second equation).
We then conclude that, for every
$\hat z=(\hat c,\theta,\eta)\in\calz_0$ the above system admits a unique solution
$\left(h^{\hat z;h_0},K^{\hat z;K_0,h_0}\right)$
such that $h^{\hat z;h_0}$ is the mild solution of the first equation with initial datum $h_0$ and $h(t)\in H_+$ for every $t\geq 0$. We can write system \eqref{eq:hK1} in a more compact way as
\begin{equation}
  \label{eq:hK2}
  \begin{cases}
    \frac{d}{dt}(\hat h,K)(t)&=
    \widetilde{A}(\hat h(t),K(t))+
    \widetilde{B}^{\hat z(t)}\left(\hat h(t),K(t)\right), \quad t \ge 0\\
    (\hat h,K)(0)&=(h_0,K_0) \in H_+ \times \R
  \end{cases}
\end{equation}
where $\widetilde{A}\colon D(A)\times\mathbb{R}\to H\times\mathbb{R}$ is the linear operator defined by
\begin{equation*}
  \widetilde{A}(h,K)=\left(Ah,-\delta K\right)
\end{equation*}
and, for $z=(c,\theta,\eta)\in Z$,
$\widetilde{B}^{z}=\widetilde{B}^{c,\theta,\eta}\colon H\times\mathbb{R}\to H\times\mathbb{R}$ is given by
\begin{equation*}
\widetilde{B}^{z}(h,K)=
\left(B^{\theta,\eta}(h),F(K,L^\theta(h))-C^c(h)-D^\eta(h)\right)\ .
\end{equation*}
The following proposition follows from basic material in Iannelli and Martcheva (2003), Iannelli (1995), Bensoussan et al. (2007).

\begin{Proposition}\label{pr:SE}
The linear operator $\widetilde{A}$ generates a strongly continuous semigroup $\widetilde{T}(t)$ on $H\times\mathbb{R}$ that leaves $H_+\times\mathbb{R}$ invariant, while the operator $B^{z}$ is Lipschitz. The Cauchy problem (\ref{eq:hK2}) admits a unique mild solution, that coincides with that of (\ref{eq:hK1}).\\
The adjoint operator of $\widetilde{A}$ with respect to the inner product $\langle (h,K),(p,Q) \rangle_{H\times\mathbb{R}}:=\langle h,p\rangle_H+KQ$, that is the linear operator $\widetilde{A}^\ast\colon \left(D(A^\ast)\times\mathbb{R}\right)\to H\times \mathbb{R}$ given by
\begin{equation*}
  \widetilde{A}^\ast(p,Q)=\left(A^\ast p,-\delta Q\right),
\end{equation*}
where
\begin{equation*}
  A^\ast=\begin{pmatrix} \frac{\partial}{\partial_a}+\mu_s&0&0\\0&\frac{\partial}{\partial_a}-\gamma&\frac{\gamma}{\pi_R^2}\\0&0&\frac{\partial}{\partial_a}+\mu_R\end{pmatrix}
\end{equation*}
on
\begin{multline*}
  D(A^\ast)=\left\{p=(p_1,p_2,p_3)\colon \frac{p_1}{\pi_S},p_2,\frac{p_3}{\pi_R^2}\in W^{1,2}(0,\bar a),\right.\\ \left.\frac{p_1}{\pi_S}(\bar a)=p_2(\bar a)=\frac{\pi_3}{\pi_R}(\bar a)=p_1(0)=0\right\}.
\end{multline*}
\end{Proposition}

Now we are in position to define precisely the set of admissible control strategies and to rewrite the target functionals. First of all, due to the presence of positivity constraints both on $\hat h$ and $K$ the set of admissible control strategies depends on the initial data and is the set
\begin{equation}\label{eq:Zdef}
  \calz_{ad}(h_0,K_0)=\left\{\hat z(\cdot)\in \calz\colon
  \left(h^{\hat z;h_0},K^{\hat z;K_0,h_0}\right)(t)\in H_+\times\mathbb{R_+}\text{ for a.e. }t\in[0,+\infty)\right\}.
\end{equation}
We now rewrite the target functionals starting by \eqref{eq:target1infhor}.
Define the function $J_1:H_+\times \mathbb{R} \times \calz_{0} \to\mathbb{R}$,
\begin{equation}\label{eq:J1def}
  J_1(h_0,K_0;\hat z)=\int_0^\infty e^{-\rho t}\int_0^{\bar a}
  \left(\hat h_1^{\hat z;h_0}(t)[a] + \hat h_2^{\hat z;h_0}(t)[a] + \hat h_2^{\hat z;h_0}(t)[a] \right)^\gamma
  u\left(\hat c(t)[a],\hat \theta(t)[a]\right)da dt.
\end{equation}
Then for every given initial datum $(h_0,K_0) \in \H_+ \times \R_+$ the problem of maximizing \eqref{eq:target1infhor} in Section 2
translate precisely in maximizing the function $J_1(h_0,K_0; \cdot)$ given \eqref{eq:J1def} over the set $\calz_{ad}(h_0,K_0)$.
It will be useful, as a shorthand, to define the function $U\colon H_+\times Z \to\mathbb{R}$,
\begin{equation}\label{eq:U1def}
U_1(h;z)=  U_1(h;c,\theta)=\int_0^{\bar a}(h_1(a)+h_2(a)+h_3(a))^\gamma u\left(c(a),\theta(a)\right)da\ ,
\end{equation}
 so that
\begin{equation}\label{eq:J1conU1def}
J_1\left(h_0,K_0;\hat z\right)
=
\int_0^\infty e^{-\rho t}
U_1\left(h^{\hat z;h_0}(t); \hat c(t),\hat \theta(t)\right)dt.
\end{equation}
The other infinite horizon problem of Section 2 have the same set of admissible strategies, hence to define it precisely it is enough the rewrite the corresponding target functional. To do this we simply have to change the function $J_1$. The target \eqref{eq:target2infhor} can be rewritten defining the functional $J_2$ in the form \eqref{eq:J1conU1def} simply substituting $U_1$
from \eqref{eq:U1def} with $U_2$ defined as follows (recall the definition of $L(\cdot)$ given in \eqref{eq:laboursupply})
\begin{equation}\label{eq:U2def}
  U_2(h,K;\theta)=F\left(K,\int_0^{\bar a}(h_1(a)+h_3(a))\alpha(a)\varphi(\theta(a))da\right)\ ,
\end{equation}
The \emph{value functions} of the two maximization problems described above are defined as
\begin{equation}\label{eq:valfuninfhor}
V_i(h,K)\colon = \sup_{\hat z\in \calz_{ad}(h,K)}J_i(h,K;\hat z),\qquad i=1,2.
\end{equation}

The other four functionals of Section 2.3 are taken with finite horizon $T>0$.
It is useful, to apply the dynamic programming approach, to let also the initial time vary.
Hence, when studying these targets the initial condition of the
state equation \eqref{eq:hK2} is taken at a generic time $t_0\in [0,T]$:
\begin{equation}
  \label{eq:hK2t0}
  \begin{cases}
    \frac{d}{dt}(\hat h,K)(t)&=
    \widetilde{A}(\hat h(t),K(t))+
    \widetilde{B}^{\hat z(t)}\left(\hat h(t),K(t)\right), \quad t \in [t_0,T]
    \\
    (\hat h,K)(t_0)&=(h_0,K_0) \in H_+ \times \R
  \end{cases}
\end{equation}
This is the state equation in this case and its solution (which exists and is unique thanks to Proposition \ref{pr:SE}) is denoted by $\left(h^{\hat z;t_0,h_0},K^{\hat z;t_0,K_0,h_0}\right)$. The set of admissible control is also a bit different:
\begin{equation}\label{eq:Zdeffinhor}
  \calz_{ad}(t_0,h_0,K_0)=\left\{\hat z(\cdot)\in \calz\colon
  \left(h^{\hat z;t_0,h_0},K^{\hat z;t_0,K_0,h_0}\right)(t)\in H_+\times\mathbb{R_+}\text{ for a.e. }t\in[t_0,T]\right\}.
\end{equation}
For the same reason also the lower extremum of the integral
of the target is taken at a generic time $t_0\in [0,T]$.
It follows that the target functional and the value function also depend on $t_0$.

To rewrite target \eqref{eq:target1finhor} we then define
\begin{equation}\label{eq:J3def}
J_3(t_0,h_0,K_0;\hat z)=
F\left(K^{\hat z;t_0,h_0,K_0}(T),
\int_0^{\bar a}\left(\hat h_1^{\hat z;t_0,h_0}(T)[a]
  + \hat h_2^{\hat z;t_0,h_0}(T)[a]
  + \hat h_3^{\hat z;t_0,h_0}(T)[a]\right) \alpha(a)da\right)
  \end{equation}
The target \eqref{eq:target2finhor} can be rewritten as
\begin{equation}\label{eq:J4def}
  \hat J_4(t_0,h_0,K_0;\hat z)=
  K^{\hat z;t_0,h_0,K_0}(T)
  \end{equation}

While the above two targets only contain a final reward, the next two contain only a current reward. To rewrite target \eqref{eq:target2finhor} we set
\begin{equation}\label{eq:J5def}
J_5\left(h_0,K_0;\hat z\right)
=
\int_0^T e^{-\rho t}
U_2\left(h^{\hat z;h_0}(t); \hat c(t),\hat \theta(t)\right)dt.
\end{equation}
where $U_2$ is defined in \eqref{eq:U2def}
Finally, to rewrite target \eqref{eq:target4finhor} we define
\begin{equation}\label{eq:J6def}
J_6(t_0,h_0,K_0;\hat z)=\int_{t_0}^T
  \int_0^{\bar a}
  \mu_I\left(a,\Xi(\hat h^{\hat z;t_0,h_0}_2(t))\right)
  \hat h^{\hat z;t_0,h_0}_2(t)[a]da dt.
\end{equation}
or simply $J_6(t_0,h_0,K_0;\hat z)=\int_{t_0}^T U_3(\hat h_2^{\hat z;t_0,h_0})dt$,
where we set
\begin{equation}\label{eq:U1def}
  U_3(h)=\int_0^{\bar a}\mu_I\left(a,\Xi(h_2)\right)h_2(a)da\ .
\end{equation}
Note that here above $\Xi:L^2(0,\bar a)\to \mathbb R$ is the linear functional
given by $\Xi(h_2)=\int_{0}^{\bar a}h_2(a)\xi(a)da$, as from \eqref{eq:Xiop}.
The \emph{value functions} in the above four finite horizon cases are defined as:
\begin{equation*}
V_i(t,h,K)\colon = \sup_{\hat z\in \calz_{ad}(t,h,K)}J_i(t,h,K;\hat z), i=3,4,5,6.
\end{equation*}

%
%

%

\color{black}

\section{Dynamic Programming and HJB equations}
\label{sec:HJB}

The starting point of the dynamic programming approach to the problems of this paper is the Dynamic Programming Principle, which we call DPP from now on, (see e.g. Theorem 1.1, p. 224 of Li and Yong, 1995, for a statement and a proof which apply to this case) which is a functional equation for the value function.
Once DPP is established the standard path is to write the differential form of DPP, the HJB equation, find a solution $v$ of it, and prove a Verification Theorem i.e. a sufficient condition for optimality in terms of the function $v$ (which can be then proved to be the value function) and its derivatives.
Both steps may be very complicated, depending on the features of the problem; this is particularly true when one deals with
problems in infinite dimension. Indeed, while
for finite dimensional problems the theory of HJB equations and of the corresponding verification results is quite well established with many regularity results, this is not the case for infinite dimensional problems. Indeed only few results are available and each case must be treated ad hoc. One can see, for example, Theorem 5.5, p.263 of Li and Yong (1995) and the papers Faggian and Gozzi (2010), Fabbri et al. (2010).

Here we abstract away from the existence and uniqueness of regular solutions of the HJB equation (which is a challenging subject and which will be next step of our work) and we concentrate on Verification Theorems and their consequences.

Since we formulated various different problems with different targets, here we concentrate on the targets \eqref{eq:target1infhor} and \eqref{eq:target4finhor} simply observing that the results can be easily extended to the other cases.

Consider first the problem of maximizing, for every initial datum $(h_0,K_0)\in \H_+\times \R_+$ the target \eqref{eq:target1infhor} over all $\hat z\in \calz_{ad}(h_0,K_0)$.
Formally, the Hamilton-Jacobi-Bellman (HJB) equation associated to such control problem is (the unknown here is $v:H_+\times \R_+\to \R$)
\begin{equation}
  \label{eq:HJB}
  \rho v(h,K)=\sup_{z\in Z}
  \mathbb{H}_{CV}(h,K,D_h v(h,K),D_K v(h,K);z)
\end{equation}
where the so-called Current Value Hamiltonian is defined as
$$
\mathbb{H}_{CV}\colon \left((D({A})\cap H_+)\times\mathbb{R}\right)\times (H\times\mathbb{R})\times Z\to\mathbb{R}
$$
\begin{equation}\label{eq:HCVdef}
    \mathbb{H}_{CV}(h,K,p,Q;z)=
    \langle \widetilde{A}(h,K),(p,Q)\rangle_{H\times\mathbb{R}}
    +\langle\widetilde{B}^{z}(h,K),(p,Q)\rangle_{H\times\mathbb{R}}
    +U_1(h;z)
\end{equation}
However this form of the HJB equation is not very convenient for two main reasons.
\begin{itemize}
  \item First of all the unknown is defined only in $H_+\times \R_+$, because $U_1$ is defined in $H_+$. This is a serious problem since the set $H_+$ has empty interior in $H$ and this creates problems in defining properly the Fr\'echet derivative $D_h v$.
      To overcome this problem we observe that $U_1$ can be immediately extended to the half space (here ${\bf 1}$ is the function with constant value $(1,1,1)$ on $H$)
\begin{equation}\label{eq:defH1}
      H_+^1:=\left\{h \in H\,\colon \; \<h,{\bf 1}\>\ge 0
      \right\}.
\end{equation}
      Indeed the interior part of this set in $H$ is simply
      $$
      Int H_+^1:=\left\{h \in H\,\colon \; \<h,{\bf 1}\> > 0
      \right\}.
      $$
      Note that in this way we are enlarging the positivity constraint on the variables $(s,i,r)$ so the resulting equation is the one associated to a different problem with a greater value function which we call $V_1^1$. However, as explained, e.g., in the appendix of Boucekkine et al. (2019), this would allow to solve also the original one if the resulting optimal strategies satisfies the original constraints (i.e. the corresponding state trajectory $\hat h $ stays in $H_+$).
  \item Second, the term $ \widetilde{A}(h,K)$ creates problem since it requires $h \in D(A)$ which is not in general satisfied when we take the mild solution $\hat{h}$ of the equation \eqref{eq:ODE}. Hence it is better to bring the operator $\widetilde{A}$ on the other side of the inner product. The drawback of this is that we need to require an additional regularity for the solution $v$: that $D_hv$ belong to $D(A^*)$ (see Definition \ref{sef:solHJB} below).
\end{itemize}

We then consider the unknown $v$ defined on $H_+^1\times \R_+$
and modify the Current Value Hamiltonian as follows (we keep the same name of it since we will be using only the following one from now on)
$$
\mathbb{H}_{CV}\colon \left(H^1_+\times\mathbb{R}\right)\times ((D(A^*)\times\mathbb{R})\times Z\to\mathbb{R}
$$
\begin{equation*}
  \begin{aligned}
    \mathbb{H}_{CV}&(h,K,p,Q;z)=
    \langle (h,K), \widetilde{A}^\ast (p,Q)\rangle_{H\times\mathbb{R}}
    +\langle\widetilde{B}^{z}(h,K),(p,Q)\rangle_{H\times\mathbb{R}}
    +U_1(h;z)\\
    &=\langle h_1, \frac{\partial p_1}{\partial a}+\mu_S p_1\rangle_{L^2_{\pi_S}}
    +\langle h_2, \frac{\partial p_2}{\partial a}-\gamma p_2+\frac{\gamma}{\pi_R^2}p_3\rangle_{L^2}
    +\langle h_3,\frac{\partial p_3}{\partial a}+\mu_R\pi_3\rangle_{L^2_{\pi_R}}
    -\delta KQ
    \\
    &-\langle \Lambda^{\theta,\eta}(h)h_1,p_1\rangle_{L^2_{\pi_S}}
    +\langle \Lambda^{\theta,\eta}(h)h_1,p_2\rangle_{L^2}
    -\langle \mu_I\left(\cdot,\Xi(h)\right)h_2, p_2\rangle_{L^2}
    +F(K,L^\theta(h))Q
    \\
    &-C^c(h)Q-D^\eta(h)Q
    +\langle(h_1+h_2+h_3)^\gamma(\cdot),u(z(\cdot))\rangle_{L^2}.
  \end{aligned}
\end{equation*}
We denote by $\mathbb{H}_{CV}^1$ the part of the Hamiltonian that depends on the controls, i.e.
 \begin{equation*}
   \begin{aligned}
     \mathbb{H}_{CV}^1(h,K,p,Q;z)&=
     -\langle \Lambda^{\theta,\eta}(h)h_1,p_1\rangle_{L^2_{\pi_S}}
     +\langle \Lambda^{\theta,\eta}(h)h_1,p_2\rangle_{L^2}
     +F(K,L^\theta(h))Q\\
     &-C^c(h)Q-D^\eta(h)Q+
     \langle(h_1+h_2+h_3)^\gamma(\cdot),u(z(\cdot))\rangle_{L^2}.
   \end{aligned}
 \end{equation*}
and set
\begin{align}\label{eq:H0def}
&\mathbb{H}^0=\mathbb{H}_{CV}-\mathbb{H}_{CV}^1
\\
\nonumber
&=\langle h_1, \frac{\partial p_1}{\partial a}+\mu_S p_1\rangle_{L^2_{\pi_S}}
    +\langle h_2, \frac{\partial p_2}{\partial a}-\gamma p_2+\frac{\gamma}{\pi_R^2}p_3\rangle_{L^2}
    +\langle h_3,\frac{\partial p_3}{\partial a}+\mu_R\pi_3\rangle_{L^2_{\pi_R}}
\\
\nonumber
&    -\delta KQ
-\langle \mu_I\left(\cdot,\Xi(h)\right)h_2, p_2\rangle_{L^2}
\end{align}
so that
\begin{multline*}
  \sup_{z\in Z} \mathbb{H}_{CV}(h,K,p,Q;z)
  =\mathbb{H}^0(h,K,p,Q)
  +\sup_{z\in Z} \mathbb{H}_{CV}^1(h,K,p,Q;z)\ .
\end{multline*}
Finally we call $\mathbb{H}^1(h,K,p,Q):=\sup_{z\in Z} \mathbb{H}_{CV}^1(h,K,p,Q;z)$
so the HJB equation \eqref{eq:HJB} rewrites as
\begin{equation}
  \label{eq:HJBnew}
  \rho v(h,K)=
  \mathbb{H}^0(h,K,p,Q)+\mathbb{H}^1(h,K,p,Q)
  \end{equation}

Now we give the definition of classical solution of \eqref{eq:HJBnew}
in the interior of our enlarged state space $H^1_+ \times \R$.
Here we abstract away from the boundary conditions as they will not be crucial for our purposes. Clearly they will become a key point when we want to prove results on existence/uniqueness/regularity of solutions of \eqref{eq:HJBnew}.

\begin{Definition}
\label{sef:solHJB}
  We say that a function
\begin{equation*}
  v\colon Int H^1_+\times (0,+\infty) \longrightarrow\mathbb{R}
\end{equation*}
is a \emph{classical solution} of the HJB equation (\ref{eq:HJB}) if
\begin{enumerate}[label=($\roman{*}$)]
\item $v$ is continuously Fr\'echet differentiable in
  $Int H_+^1\times (0,+\infty) $;
\item the derivative $D_hv(h,K)$ belongs do $D(A^\ast)$ for every $(h,K)\in Int H^1_+\times (0,+\infty) $ and $A^\ast D_hv$ is continuous in $Int H^1_+\times (0,+\infty) $;
\item $v$ satisfies equation (\ref{eq:HJB}) for every $(H,k)\in Int H_+^1\times (0,+\infty) $.
\end{enumerate}
\end{Definition}

We have the following result, which generalizes, e.g., Proposition 1.2, p. 225 of Li and Yong (1995).

\begin{Theorem}\label{th:V}
Consider the problem of optimizing the target functional \eqref{eq:J1def} over the set of control strategies
\begin{equation}\label{eq:Z1def}
  \calz^1_{ad}(h_0,K_0)=\left\{\hat z(\cdot)\in \calz\colon
  \left(h^{\hat z;h_0},K^{\hat z;K_0,h_0}\right)(t)\in H^1_+\times\mathbb{R_+}\text{ for a.e. }t\in[0,+\infty)\right\},
\end{equation}
where $H^1_+$ is defined as in \eqref{eq:defH1}.
Suppose that the value function $V^1_1$ of this ``enlarged'' problem is continuously Fr\'echet differentiable in $Int H^1_+\times (0,+\infty)$ and that $D_hV(h,K)\in D(A^\ast)$ for every $(h,K)\in Int H^1_+\times (0,+\infty)$. Then $V$ is a classical solution of the Hamilton-Jacobi-Bellman equation (\ref{eq:HJB})
in $Int H^1_+\times (0,+\infty)$ .
\end{Theorem}
\begin{proof}
For simplicity, in this proof, we will write $V$ for $V^1_1$ $\nabla V$ for the vector $(D_hV^1_1, D_KV^1_1)$.
By Theorem 1.1, p. 224 of Li and Yong (1995) $V$ satisfies the dynamic programming principle, that is, for every $(h_0,K_0)\in
Int H^1_+\times (0,+\infty)$ and every $t\geq 0$
  \begin{equation}\label{eq:DPP}
    V(h_0,K_0)=\sup_{\hat z\in\calz^1_{ad}(h_0,K_0)}
    \left\{\int_0^te^{-\rho s}
    U_1\left(\hat h^{\hat z;h_0}(s),\hat z(s)\right)ds
    +e^{-\rho t}
    V\left(\hat h^{\hat z;h_0}(t),K^{\hat z;h_0,K_0}(t)\right)\right\}
  \end{equation}
From now on for simplicity we will write $\underline{h}^{\hat z}(t)$
for $\left(\hat h^{\hat z;h_0}(t),K^{\hat z;h_0,K_0}(t)\right)$.
Using the chain rule for mild solutions (see for example Proposition 5.5 Li and Yong, 1995) we have, for every
$(h_0,K_0)\in Int H^1_+\times (0,+\infty)$,
every $\hat z=(\hat c,\hat \theta,\hat \eta)\in\calz^1_{ad}(h_0,K_0)$
 and every $t\geq 0$,
  \begin{multline*}
    V\left(\underline{h}^{\hat z}(t)\right)
    -V(h_0,K_0)
    =\int_0^t
\langle \underline{h}^{\hat z}(s),
\widetilde{A}^\ast \nabla V(\underline{h}^{\hat z}(s))
\rangle_{H\times\mathbb{R}}ds\\
+\int_0^t\langle
\widetilde{B}^{\hat z} (\underline{h}^{\hat z}(s)),
\nabla V(\underline{h}^{\hat z}(s))
\rangle_{H\times\mathbb{R}}ds.
  \end{multline*}
Therefore, using also \eqref{eq:DPP},
\begin{align*}
  0&\geq \int_0^te^{-\rho s}
  U_1(h^{\hat z;h_0}(s);\hat c(s),\hat \theta(s))ds
  +e^{-\rho t}
  V\left(\underline{h}^{\hat z}(t)\right)-V(h_0,K_0)
  \\
   &
   =\int_0^te^{-\rho s}
   U_1(h^{\hat z;h_0}(s);\hat c(s),\hat \theta(s))ds
   +e^{-\rho t}
   V\left(\underline{h}^{\hat z}(t)\right)
   -e^{-\rho t} V(h_0,K_0)
   +(e^{-\rho t}-1)V(h_0,K_0)
   \\
   &
   =\int_0^t e^{-\rho s}
   U_1(h^{\hat z;h_0}(s);\hat c(s),\hat \theta(s))ds
   +e^{-\rho t}\int_0^t \langle
   \underline{h}^{\hat z}(s),
   \widetilde{A}^\ast\nabla V(\underline{h}^{\hat z}(s))
   \rangle_{H\times\mathbb{R}}ds
   \\
   &
   +\int_0^t\langle
   \widetilde{B}^{\hat z}(\underline{h}^{\hat z}(s)),
    \nabla V(\underline{h}^{\hat z}(s))
    \rangle_{H\times\mathbb{R}}ds
    + (e^{-\rho t}-1)V(h_0,K_0).
\end{align*}

Now, since we are in an open set we know that the control strategies can be taken constant (so $\hat z(t)=\hat z(0)$ for all $t\ge 0$ for a while.
We now divide both sides of the inequality by $t$ and take the limit as $t\to 0$; finding
\begin{multline}
  0\geq U(h_0;\hat c(0),\hat \theta(0))
  +\langle (h_0,K_0),
  \widetilde{A}^\ast \nabla V(h_0,K_0)
  \rangle_{H\times\mathbb{R}}\\
  +\langle
  \widetilde{B}^\pi(h_0,K_0),
    \nabla V(h_0,K_0)
\rangle_{H\times\mathbb{R}}
-\rho V(h_0,K_0).
\end{multline}
Therefore we obtain
\begin{equation*}
  0\geq \sup_{z\in Z}\mathbb{H}_{CV}
  (h_0,K_0,D_hV(h_0,K_0),D_KV(h_0,K_0);z)-\rho V(h_0,K_0).
\end{equation*}
To prove the reverse inequality we fix again
$(h_0,K_0)\in Int H^1_+\times (0,+\infty)$;
by definition of the value function, for any positive $\epsilon$ and any positive $t$ we can find an admissible control
$\hat z^\epsilon(\cdot)\in\calz^1_{ad}(h_0,K_0)$ such that
\begin{multline*}
  -\epsilon t\leq
  \int_0^te^{-\rho s}
  U_1\left(\underline{h}^{\hat z^\epsilon}(s);\hat c^\epsilon(s),\hat \theta^\epsilon(s)\right)ds
  +e^{-\rho t} V(\underline{h}^{\hat z^\epsilon}(t))-V(h_0,K_0)\ .
\end{multline*}
Using the equation satisfied by $\underline{h}^{\\hat z^\epsilon}$ we get
\begin{align*}
  -\epsilon t&\leq e^{-\rho t}
   \left(V(\underline{h}^{\hat z^\epsilon}(t)-V(h_0,K_0)\right)\\
& \qquad   \qquad\qquad\qquad
   +\int_0^te^{-\rho s}
   U_1\left(h^{\hat z^\epsilon;h_0}(s);
   \hat c^\epsilon(s),\hat \theta^\epsilon(s)\right)ds
   +(e^{-\rho t}-1)V(h_0,K_0)
   \\
&
=e^{-\rho t}\langle
\widetilde{T}(t)(h_0,K_0)-(h_0,K_0),
\nabla V(h_0,K_0)
\rangle_{H\times\mathbb{R}}
\\
&
+e^{-\rho t}\langle
\int_0^t\widetilde{T}(t-s)\widetilde{B}^{\pi^\epsilon}
(\underline{h}^{\pi^\epsilon}(s)),
\nabla V(h_0,K_0)
\rangle_{H\times\mathbb{R}}ds
\\
&
+\int_0^te^{-\rho s}
U_1\left(h^{\hat z^\epsilon;h_0}(s); \hat c^\epsilon(s),\hat \theta^\epsilon(s)\right)ds
+o(t)+(e^{-\rho t}-1)V(h_0,K_0)
\end{align*}
We can then find a continuous function $\sigma:[0,+\infty]\to[0,+\infty]$ such that $\sigma(0)=0$ and
\begin{align*}
\sigma(\epsilon)&\leq
\frac1t e^{-\rho t}\langle
(\widetilde{T}(t)-\mathrm{Id})(h_0,K_0),
\nabla V(h_0,K_0),
\rangle_{H\times\mathbb{R}}
+\frac1t \int_0^t\langle
\widetilde{B}^{\pi^\epsilon}(h_0,K_0),
\nabla V(h_0,K_0)
\rangle_{H\times\mathbb{R}}ds
\\
&
+\frac1t \int_0^te^{-\rho a}
U_1\left(h_0;\hat c^\epsilon(s),\hat \theta^\epsilon(s)\right)ds
+o(1)
+\frac{e^{-\rho t}-1}{t}V(h_0,K_0)
\\
&
\leq
\sup_{z\in Z}\mathbb{H}_{CV}
(h_0,K_0,D_h V(h_0,K_0),D_K V(h_0,K_0);z)
+\frac{e^{-\rho t}-1}{t}V(h_0,K_0)+o(1),
\end{align*}
which implies, taking the limit as $t\to 0$,
\begin{equation*}
  \sigma(\epsilon)\leq \sup_{z\in Z}
  \mathbb{H}_{CV}(h_0,K_0,D_h V(h_0,K_0),D_K V(h_0,K_0);z)
  -\rho V(h_0,K_0).
\end{equation*}
Letting now $\epsilon$ go to $0$ we get the result.
\end{proof}

\begin{Remark}
The above Theorem \eqref{th:V} holds in a completely similar way for the other problems where the target is changed.
Of course, in case of finite horizon problems the HJB equation is different and, for example, in the case of target \eqref{eq:J4def}, is
$$
-\frac{\partial v(t,h,K)}{\partial t}=
\sup_{z\in Z}\mathbb{H}_{CV}
\left(h,K,D_h v(h,K),D_K v(h,K);z\right)
$$
for $t\in[0,T]$, $(h,K)\in Int H^1_+ \times (0,+\infty)$
and with the final condition $v(T,h,K)=K$.

Note finally that in all such cases we would take the enlarged constraint
$h\in H^1_+$ instead of the one $h \in H_+$.
\end{Remark}

\section{Verification theorems}
\label{sec:verification}

We first recall the definition of optimal strategy for our starting problem and for the ``enlarged'' one.

\begin{Definition}
  For $(h_0,K_0)\in H_+\times\mathbb{R}_+$
  (respectively $(h_0,K_0)\in Int H^1_+\times (0,+\infty)$), an admissible control strategy $\hat z^\ast \in\calz_{ad}(h_0,K_0)$   (respectively $\hat z^\ast \in\calz^1_{ad}(h_0,K_0)$) is called \emph{optimal} at $(h_0,K_0)$ if
  \begin{equation*}
    V_1(h_0,K_0)=J_1(h_0,K_0;\hat z^\ast(\cdot)),
    \hbox{ respectively }
V_1^1(h_0,K_0)=J_1(h_0,K_0;\hat z^\ast(\cdot)),
  \end{equation*}
that is, if it is a maximizer for $J$. The corresponding solution $(h^{\pi^\ast;h_0},K^{\pi^\ast;K_0})$ of (\ref{eq:hK2}) is called an \emph{optimal state trajectory}.
\end{Definition}
The following result is the so-called Verification Theorem
which provides sufficient optimality conditions.
\begin{Theorem}
\label{thm:comparison}
  Let $v$ be a classical solution of the HJB equation \eqref{eq:HJB}, with the additional property that for every $\hat z\in\calz^1_{ad}(h_0,K_0)$
  \begin{equation}
    \label{eq:lim_0}
    \lim_{T\to+\infty}e^{-\rho T}
    v\left(\hat h^{\hat z;h_0}(T),K^{\hat z;h_0,K_0}(T)\right)=0;
  \end{equation}
then $V_1^1(h_0,K_0)\leq v(h_0,K_0)$ for every $(h_0,K_0)\in Int H^1_+\times (0,+\infty)$.
Moreover, if an admissible control $\hat z^*\in \calz^1_{ad}(h_0,K_0)$ is such that
\begin{align}\label{eq:hamsup}
\sup_{z\in Z}
\mathbb{H}_{CV}
(\underline{h}^{\hat z^*}(t),\nabla v(\underline{h}^{\hat z^*}(t)),z)
=
\mathbb{H}_{CV}
(\underline{h}^{\hat z^*}(t),\nabla v(\underline{h}^{\hat z^*}(t)),
\hat z^*(t))
\end{align}
then $\hat z^*$ is optimal at $(h_0,K_0)$ and
$V_1^1(h_0,K_0)=v(h_0,K_0)$.
\end{Theorem}
\begin{proof}
We write $\nabla v$ for $(D_h v,D_K v)$.
Moreover, as in the proof of Theorem 4.2 we write for simplicity $\underline{h}^{\hat z}(t)$ in place of
$(h^{\hat z;h_0}(t),K^{\hat z;h_0,K_0}(t))$.
We first prove that, for every $\hat z\in \calz^1_{ad}(h_0,K_0)$ we have the fundamental identity
  \begin{multline}
    \label{eq:fundamental}
    v(h_0,K_0)
    =
    J_1(h_0,K_0;\hat z(\cdot))
    \\
+\int_0^\infty e^{-\rho t}
\left[\sup_{z\in Z}
\mathbb{H}_{CV}
(\underline{h}^{\hat z}(t),\nabla v(\underline{h}^{\hat z}(t)),z) \right.
\\
\left.
-\mathbb{H}_{CV}
(\underline{h}^{\hat z}(t),\nabla v(\underline{h}^{\hat z}(t)),
\hat z(t))\right]dt.
  \end{multline}
Indeed, differentiating the function
$t\mapsto e^{-\rho t}v(\underline{h}^{\hat z}(t))$ and integrating on $[0,T]$ we find
\begin{align*}
v(h_0,K_0)
&
=e^{-\rho T} v(\underline{h}^{\hat z}(T))
+\int_0^Te^{-\rho t}\rho v(\underline{h}^{\hat z}(t)) dt
\\
&
-\int_0^T e^{-\rho t}\langle
(\underline{h}^{\hat z}(t)),
\widetilde{A}^\ast
\nabla v(\underline{h}^{\hat z}(t))
\rangle_{H\times\mathbb{R}}dt
\\
&
-\int_0^T e^{-\rho t}\langle
\widetilde{B}^{{\hat z}(t)}(\underline{h}^{\hat z}(t)),
\nabla v(\underline{h}^{\hat z}(t))
\rangle_{H\times\mathbb{R}}dt.
\end{align*}
We can then add and subtract the term
$\int_0^T e^{-\rho t} U_2(h^{{\hat z},h_0}(t),\hat z(t))dt$ on the right hand side and use the fact that $v$ solves the HJB equation \eqref{eq:HJB} to obtain
\begin{align*}
v(h_0,K_0)
&
=e^{\rho T}v(\underline{h}^{\hat z}(T))
+\int_0^T U_2(h^{{\hat z};h_0}(t), \hat z(t))dt
\\
&
+\int_0^T e^{-\rho t}\left[\sup_{z\in Z}
\mathbb{H}_{CV}
(\underline{h}^{\hat z}(t),\nabla v(\underline{h}^{\hat z}(t));
\hat z(t)) \right.
\\
&
\left.
-\mathbb{H}_{CV}(\underline{h}^{\hat z}(t),
\nabla v(\underline{h}^{\hat z}(t)); \hat z(t))\right]dt.
\end{align*}
The fundamental identity then follows taking the limit as $T\to\infty$ and using (\ref{eq:lim_0}). Since the last integral is always non-negative, and eventually taking the supremum over all admissible controls on the right hand side, we get the first claim.
The second claim follows observing that, for such $\hat z^*$ we have, from \eqref{eq:fundamental} and the first claim,
$$
V_1^1(h_0,K_0)\le v(h_0,K_0)=J_1(h_0,K_0;\hat z^*)
$$
which implies that $V_1^1(h_0,K_0)=v(h_0,K_0)=J_1(h_0,K_0;\hat z^*)$
and so the claim.
 \end{proof}
\begin{Corollary}
\label{cor:verification1}
Let $v$ be a classical solution of the HJB equation \eqref{eq:HJB}
and assume that the set valued map
\begin{equation*}
    (h,K) \to \argmax_{z\in Z}
    \mathbb{H}_{CV}
    \left(h,K,\nabla v(h,K);z\right)
  \end{equation*}
admits a measurable selection $G: Int H^1_+\times (0,+\infty)\to Z$.
Let $(h_0,K_0)\in Int H^1_+\times (0,+\infty)$
and assume that the closed loop equation
\begin{equation}
  \label{eq:hK2closedloop}
  \begin{cases}
    \frac{d}{dt}(\hat h,K)(t)&=
    \widetilde{A}(\hat h(t),K(t))+
    \widetilde{B}^{G(\hat h(t),K(t))}\left(\hat h(t),K(t)\right), \quad t \ge 0\\
    (\hat h,K)(0)&=(h_0,K_0),
  \end{cases}
\end{equation}
admits a solution $(\hat h^{G;h_0},K^{G;h_0,K_0})$ such that
the control strategy
$$
\hat z^*(t)=G\left(\hat h^{G;h_0}(t),K^{G;h_0,K_0}(t)\right)
$$
belongs to $\calz_{ad}^1(h_0,K_0)$.
then $\hat z^*$ is optimal.
\end{Corollary}
\begin{proof}
It immediately follows from the previous Theorem  \ref{thm:comparison} and from the fundamental identity (\ref{eq:fundamental}).
\end{proof}

\begin{Corollary}
\label{cor:verification2}
Suppose that the value function $V_1^1$ is a classical solution of the Hamilton-Jacobi-Bellman equation \eqref{eq:HJB}, that $\hat z^\ast \in\calz^1_{ad}(h_0,K_0)$ is optimal at $(h_0,K_0)$ and that
  \begin{equation*}
    \lim_{T\to +\infty}e^{-\rho T}V_1^1\left(h^{\pi^\ast;h_0}(t),K^{\pi^\ast;K_0}(t)\right)=0\ .
  \end{equation*}
Then $\hat z^*$ satisifes \eqref{eq:hamsup}.
\end{Corollary}
\begin{proof}
  Proceeding as in the proof of Theorem \ref{thm:comparison} we can show that the value function $V$ satisfies the fundamental identity (\ref{eq:fundamental}). Since $\hat z^\ast(\cdot)$ is optimal the integral term on the right hand side of (\ref{eq:fundamental}) must be $0$, and the claim follows.
\end{proof}

\begin{Remark}
The above results allow, if we can find $V_1^1$, at least numerically, to solve the problems with the enlarged constraints.
To pass to our initial control problem we have to show that, for some $(h_0,K_0) \in H_+\times \R_+$, the optimal control of the enlarged problem is als admissible for the initial problem.
This has been done e.g. in Boucekkine et al (2019) and the same idea may work in some special cases of our set-up.
\end{Remark}

\section{Conclusion}
\label{sec:conclusion}

Given the strong differences in the effects of some epidemics (and particularly that of COVID-19) as individuals vary in age, it is important, in trying to understand the economic impact of the contagion and in evaluating the policies to combat it, to model it as precisely as possible.

In the previous contributions which integrate the epidemiological dynamics in macro-dynamic models, the stratification by age of population is often absent and, when introduced, it is modeled using a finite number of groups with no possibility to move from one group to another.

In this paper we propose a general fully age-structured time continuous set-up for macro analysis of epidemics and economic dynamic.

After rewriting the problem using a suitable Hilbert space reformulation of the associated infinite dimensional optimal control problem, we provide verification type results which, given our general infinite dimensional setting cannot be derived from previous results in the literature.



\bigskip




\bigskip

\bigskip

\section*{Acknowledgements}
The work of Giorgio Fabbri is supported by the French National Research Agency in the framework of the ``Investissements d'avenir'' program (ANR-15-IDEX-02) and in that of the center of excellence LABEX MME-DII (ANR-11-LABX-0023-01).

The work of Fausto Gozzi and Giovanni Zanco is supported by the Italian Ministry of University and Research (MIUR),  in the framework of PRIN projects 2015233N54\_006 (“Deterministic and stochastic evolution equations”) and 2017FKHBA8\_001 (“The Time-Space Evolution of Economic Activities: Mathematical Models and Empirical Applications”)

\bigskip

\bigskip

\bigskip

\bigskip

\end{document}